\documentclass[aps,superscriptaddress,twocolumn,twoside,floatfix,prl,a4paper,longbibliography]{revtex4-1}

\usepackage{times}
\usepackage{epsfig}
\usepackage{amsfonts}
\usepackage{amsmath}
\usepackage{amssymb,amsthm}
\usepackage{mathtools}
\usepackage{xcolor,colortbl}
\usepackage{multirow}
\usepackage{braket}
\usepackage{latexsym}
\usepackage{natbib}
\usepackage{verbatim}
\usepackage{gensymb}

\usepackage{caption}
\usepackage{subcaption}
\usepackage{ragged2e}
\DeclareCaptionJustification{justified}{\justifying}
\usepackage{blkarray}
\usepackage{graphicx}

\newtheorem{theorem}{Theorem}
\newtheorem{corollary}[theorem]{Corollary}

\newtheorem{proposition}[theorem]{Proposition}

\newtheorem{Lemma}[theorem]{Lemma}

\usepackage{newpxtext,newpxmath}

\let\coloneqq\relax

\usepackage[utf8]{inputenc}
\usepackage{amsthm}
\usepackage{amssymb}
\usepackage{amsmath}
\usepackage{bbold}
\usepackage{bbm}
\usepackage[pdftex, backref=page]{hyperref}
\hypersetup{
    colorlinks=true, 
    linkcolor=blue,  
    citecolor=blue!40!green,
    filecolor=magenta,      
    urlcolor=cyan,
}
\usepackage{braket}
\usepackage{dsfont}
\usepackage{mathdots}
\usepackage{mathtools}
\usepackage{enumerate}
\usepackage[shortlabels]{enumitem}
\usepackage{csquotes}
\usepackage{stmaryrd}
\usepackage[cal=boondox]{mathalfa}
\usepackage{graphicx}
\usepackage{stackengine}
\usepackage{scalerel}
\usepackage{tensor}       
\usepackage{array}
\usepackage{makecell}
\newcolumntype{x}[1]{>{\centering\arraybackslash}p{#1}}
\usepackage{tikz}
\usepackage{pgfplots}
\usetikzlibrary{shapes.geometric, shapes.misc, positioning, arrows, arrows.meta, decorations.pathreplacing, decorations.pathmorphing, patterns, angles, quotes, calc}
\usepackage{booktabs}
\usepackage{xfrac}
\usepackage{siunitx}
\usepackage{centernot}
\usepackage{comment}
\usepackage{chngcntr}
\usepackage{caption}
\usepackage{subcaption}

\newtheorem{thm}{Theorem}
\newtheorem*{thm*}{Theorem}

\newtheorem*{prop*}{Proposition}

\newtheorem*{lemma*}{Lemma}

\newtheorem*{cor*}{Corollary}

\newtheorem*{cj*}{Conjecture}

\newtheorem*{Def*}{Definition}

\newtheorem*{question*}{Question}

\newtheorem*{problem*}{Problem}

\makeatletter
\def\thmhead@plain#1#2#3{%
  \thmname{#1}\thmnumber{\@ifnotempty{#1}{ }\@upn{#2}}%
  \thmnote{ {\the\thm@notefont#3}}}
\let\thmhead\thmhead@plain
\makeatother

\theoremstyle{definition}
\newtheorem{rem}[thm]{Remark}

\newtheorem{manualthminner}{Theorem}
\newenvironment{manualthm}[1]{%
  \renewcommand\themanualthminner{#1}%
  \manualthminner \it
}{\endmanualthminner}

\newcommand{\bb}{\begin{equation}\begin{aligned}\hspace{0pt}}
\newcommand{\bbb}{\begin{equation*}\begin{aligned}}
\newcommand{\ee}{\end{aligned}\end{equation}}
\newcommand{\eee}{\end{aligned}\end{equation*}}

\newcommand{\ketbra}[1]{\ket{#1}\!\!\bra{#1}}

\newcommand{\ketbraa}[2]{\ket{#1}\!\!\bra{#2}}

\renewcommand{\epsilon}{\varepsilon}

\newcommand{\id}{\mathds{1}}

\newcommand{\onelocc}{\mathrm{LOCC}_\to}
\newcommand{\locc}{\mathrm{LOCC}}
\newcommand{\sep}{\mathrm{SEP}}
\newcommand{\ppt}{\mathrm{PPT}}

\newcommand{\sepp}{\mathrm{NE}}

\newcommand{\kp}{\mathrm{KP}}

\newcommand{\ane}{\mathrm{ANE}}

\newcommand{\Esq}{E_{\text{\textit{sq}}}}

\DeclareMathAlphabet{\pazocal}{OMS}{zplm}{m}{n}

\DeclareMathOperator{\tr}{tr}

\newcommand{\FF}{\pazocal{F}}

\newcommand{\lsmatrix}{\left(\begin{smallmatrix}}
\newcommand{\rsmatrix}{\end{smallmatrix}\right)}

\stackMath

\stackMath

\makeatletter
\newcommand*\rel@kern[1]{\kern#1\dimexpr\macc@kerna}
\newcommand*\widebar[1]{%
  \begingroup
  \def\mathaccent##1##2{%
    \rel@kern{0.8}%
    \overline{\rel@kern{-0.8}\macc@nucleus\rel@kern{0.2}}%
    \rel@kern{-0.2}%
  }%
  \macc@depth\@ne
  \let\math@bgroup\@empty \let\math@egroup\macc@set@skewchar
  \mathsurround\z@ \frozen@everymath{\mathgroup\macc@group\relax}%
  \macc@set@skewchar\relax
  \let\mathaccentV\macc@nested@a
  \macc@nested@a\relax111{#1}%
  \endgroup
}

\counterwithin*{equation}{part}
\counterwithin*{thm}{part}
\counterwithin*{figure}{part}

\tikzset{meter/.append style={draw, inner sep=10, rectangle, font=\vphantom{A}, minimum width=30, line width=.8, path picture={\draw[black] ([shift={(.1,.3)}]path picture bounding box.south west) to[bend left=50] ([shift={(-.1,.3)}]path picture bounding box.south east);\draw[black,-latex] ([shift={(0,.1)}]path picture bounding box.south) -- ([shift={(.3,-.1)}]path picture bounding box.north);}}}
\tikzset{roundnode/.append style={circle, draw=black, fill=gray!20, thick, minimum size=10mm}}
\tikzset{squarenode/.style={rectangle, draw=black, fill=none, thick, minimum size=10mm}}

\definecolor{Blues5seq1}{RGB}{239,243,255}
\definecolor{Blues5seq2}{RGB}{189,215,231}
\definecolor{Blues5seq3}{RGB}{107,174,214}
\definecolor{Blues5seq4}{RGB}{49,130,189}
\definecolor{Blues5seq5}{RGB}{8,81,156}

\definecolor{Greens5seq1}{RGB}{237,248,233}
\definecolor{Greens5seq2}{RGB}{186,228,179}
\definecolor{Greens5seq3}{RGB}{116,196,118}
\definecolor{Greens5seq4}{RGB}{49,163,84}
\definecolor{Greens5seq5}{RGB}{0,109,44}

\definecolor{Reds5seq1}{RGB}{254,229,217}
\definecolor{Reds5seq2}{RGB}{252,174,145}
\definecolor{Reds5seq3}{RGB}{251,106,74}
\definecolor{Reds5seq4}{RGB}{222,45,38}
\definecolor{Reds5seq5}{RGB}{165,15,21}

\allowdisplaybreaks

\usepackage[most,breakable]{tcolorbox}
{\expandafter\ifstrequal\expandafter{#1}{orange}{\begin{tcolorbox}[colback=red!15,colframe=orange!15,breakable,enhanced]}{\begin{tcolorbox}[colback=Blues5seq1,colframe=Blues5seq5,breakable,enhanced]}}%
{\end{tcolorbox}}

\DeclareMathAlphabet{\cleancal}{OMS}{cmsy}{m}{n}

\makeatletter
\newcommand{\nocontentsline}[3]{}
\let\origcontentsline\addcontentsline
\newcommand\stoptoc{\let\addcontentsline\nocontentsline}
\newcommand\resumetoc{\let\addcontentsline\origcontentsline}
\makeatother

\begin{document}  

\addtocontents{toc}{\protect\setcounter{tocdepth}{-1}}
\title{Entanglement of flower states}

\author{Samrat Sen}
\affiliation{Scuola Normale Superiore, Piazza dei Cavalieri 7, 56126 Pisa, Italy}

\author{Ludovico Lami}
\affiliation{Scuola Normale Superiore, Piazza dei Cavalieri 7, 56126 Pisa, Italy}

\begin{abstract}
The mysterious nature of entanglement, one of the most prominent exquisitely quantum phenomena, is reflected in its intricate operational structure, with a hierarchy of classes of free operations that enable its manipulation at different levels of effectiveness. Here we use the class of `flower states', parametrised by their (even) local dimension $2k$, to shine light on some aspects of this varied landscape. We compute all the main entanglement measures for flower states, uncovering a large gap between all forms of distillable entanglement, equal to $1$ ebit independently of the local dimension, and the entanglement cost under local operations and classical communication (LOCC), known to be equal to $\log\big(2\sqrt{k}\big)$. Even under the strictly more powerful class of non-entangling (NE) operations, we show that their cost is still equal to $\log\big(1+\sqrt{k}\big)$, only about an ebit less than for LOCCs. This result, which we prove by calculating the recently introduced tempered entanglement negativity for these states, demonstrates the largest known `irreversibility gap', i.e.\ the difference between distillable entanglement and entanglement cost, under NE operations, equal to $\Theta\big(\frac12 \log d\big)$, with $d$ being the local dimension. A notable consequence is that the celebrated squashed entanglement is not a monotone under NE operations. Finally, we compute the exact cost under LOCC operations for flower states; this is given by the Schmidt number, which turns out to be additive over multiple copies and equal to $\min_{r|k} \log\left( r + \frac{k}{r} \right)$; for prime $k$ this reduces to $\log(k+1)$, about twice the standard LOCC cost. These last results leverage the uncertainty relations over cyclic groups proved by Tao and Meshulam.

\end{abstract}
\maketitle

\stoptoc

\section{Introduction}


Quantum entanglement is the foundational resource underpinning the advantages of quantum information theory, such as communication~\cite{bennett92,bennett93,bennett96}, communication complexity~\cite{Buhrman2001,Brukner2004}, measurement precision~\cite{Giovannetti2004}, computation~\cite{rauss01} and cryptography~\cite{ekert91}.  
Its manipulation is governed by free operations that reflect fundamental laboratory constraints. While local operations and classical communication (LOCC) provide the natural physical framework for this setting, 
    their intricate mathematical structure renders precise characterisations under LOCC notoriously intractable~\cite{Chitambar2014}. To circumvent these hurdles and, at the same time, study settings in which Alice and Bob are granted additional resources, several alternative sets of free operations have been investigated. 
    For instance, restricting the framework to one-way LOCC protocols—which permit only a single round of communication—establishes rigorous lower bounds on crucial entanglement properties. A prominent example is the hashing  bound~\cite{hashing}, which provides a fundamental lower bound for the number of ebits that can be distilled from many copies of a given quantum state.
	
	While restricting protocols to one-way LOCC provides an inner approximation that establishes rigorous lower bounds, one can alternatively relax the LOCC constraints 
    so as to obtain complementary upper bounds. PPT operations, introduced in the seminal works of Rains~\cite{rains99, rains01}, and separability-preserving or non-entangling operations (NE)~\cite{Brando2010} are paradigmatic examples of this top-down 
    approach.
	
	The computations do become simpler under these approximations, allowing us to derive much-needed, tight upper bounds on entanglement manipulation. However, 
    that is hardly the end of the story. The landscape of entanglement theory itself gets reshaped because of these approximations. For instance, bound entanglement, 
    a signature feature of the theory under LOCC operations, is absent under non-entangling (NE) 
    operations~\cite{Lami2024}, meaning that relaxing LOCC to NE unlocks 
    new distillation protocols and 
    enhances the power of entanglement distillation. The permissive power, naturally, has an impact on the dilution cost as well. A paradigmatic example of this is the antisymmetric state, dubbed the ‘universal counterexample’ in entanglement theory~\cite{counter}: under LOCC, it exhibits extreme irreversibility, with its distillable entanglement scaling as $\cleancal{O}(1/d)$ while its entanglement cost remains bounded below by a dimension-independent positive constant~\cite{Christandl2012}; however, the same state becomes perfectly reversible under PPT operations~\cite{audenaert2003entanglement}, owing to a sharp drop in its entanglement cost. 
	
 However, this painting of the landscape of entanglement theory is far from complete. To try to shed light on some of the many remaining blind spots, we need to expand our zoology of examples, investigating classes of quantum states with notable properties.
    In our work, we utilise the family of `flower states', a class of maximally correlated states parametrised by their even local dimension $d = 2k$, to systematically explore this intricate operational landscape. It is exceedingly rare to encounter non-trivial states for which each and every entanglement measure can be explicitly computed, and one of our main conceptual contributions is to show that flower states belong to this exclusive group.
	
	To begin with, we establish that the NE entanglement cost is $\frac{1}{2}\log d + o(\log d)$. 
    Conversely, the distillable entanglement remains exactly 1 ebit regardless of local dimension, a feature holding for both flower states and their generalisations (`generalised flower states'). Crucially, we uncover an explicit one-shot, one-way LOCC protocol that deterministically distils this single ebit, proving that the permissive power of NE operations is entirely redundant for distillation here. As an immediate consequence, we obtain a new bound on a fundamental constant characterising entanglement theory, the NE irreversibility gap $\Delta_{\mathrm{irr}}^{\sepp}$, which quantifies the largest achievable difference between the NE entanglement cost and the NE distillable entanglement, divided by the number of local qubits. Prior estimates yielded $\Delta_{\mathrm{irr}}^{\sepp} \geq \frac{2}{\log_2 3}-1 \approx 0.262$, while flower states provide the much stronger bound $\Delta_{\mathrm{irr}}^{\sepp} \geq 1/2$.   
    
	Furthermore, we address the role of squashed entanglement ($\Esq$)~\cite{tucci,christandl2005uncertainty} within this expanded operational landscape. Under LOCC, $\Esq$ famously bounds the operational rates, satisfying $E_d \leq \Esq \leq E_c$. Its inherent mathematical structure, such as additivity~\cite{Christandl2004}, monogamy~\cite{winter2004}, faithfulness~\cite{Brando2011} and asymptotic continuity~\cite{Alicki_2004}, makes it highly appealing. 
       Since NE operations are vastly more permissive than LOCC~\cite{Chitambar2014}, establishing a robust lower bound for the NE entanglement cost is a major structural challenge, and it was open whether squashed entanglement could be used in this way. Here, we answer this question in the negative: \emph{the squashed entanglement is not a monotone under NE operations}; 
    in fact, it can be strictly larger than the NE entanglement cost. Far from being a negative result, this actually demonstrates how the squashed entanglement is inextricably tailored to the physical constraints of local actions, establishing it as a strict structural separator between LOCC and its axiomatic outer approximations.
	
	Moving beyond the standard vanishing-error regime, into one with a strict zero-error constraint in the many-copy limit, one naturally expects the entanglement cost to increase. Surprisingly, we show that this is not the case for flower states under non-entangling operations, as the exact NE cost turns out to be equal to its corresponding vanishing-error counterpart. 
    
    In contrast, within the exact LOCC paradigm,  
    the entanglement cost exhibits a much richer structure.
   We reveal this through a detailed characterisation of the flower states and comparing it with 
  the standard LOCC cost. For the flower states, Christandl and Winter~\cite{christandl2005uncertainty} showed that the standard LOCC cost coincides exactly with the squashed entanglement, making these states the basis for two significant firsts: the first exact computation of squashed entanglement, and the first demonstration that this measure displays the peculiar effect of `locking'.

  In the zero-error setting, the exact entanglement cost is known to be given by the regularised Schmidt number, which is in general extremely difficult to calculate, due to submultiplicativity effects.
  Remarkably, we overcome this barrier for flower states by deriving an exact analytical expression for this quantity, proving that it is strictly additive over multiple copies and given by $\min_{r|k} \log\left( r + \frac{k}{r} \right)$, where the minimisation is over divisors of $k$. For prime $k$, this expression elegantly reduces to $\log(k+1)$, representing an approximate twofold increase over the standard LOCC cost. This demonstrates
  the asymptotically largest possible gap between exact LOCC dilution and distillation: the cost is $\log (1+d/2) \approx \log d$, while the distillable entanglement is $1$.  On the contrary, when $k$ is a perfect square, the exact LOCC cost coincides with its standard counterpart, and both equal $\log\big(2\sqrt{k}\big)$. These 
  results are rigorously established by leveraging the uncertainty relations over cyclic groups introduced by Tao~\cite{Tao2005} and Meshulam~\cite{meshulam}, thereby completing a comprehensive mapping of the entanglement properties of flower states across these varied operational regimes.

\section{Preliminaries}

 For a given state  $\rho$ acting on $\mathbb{C}^d$, we define the corresponding maximally correlated state $\Omega_{\rho}$, acting on  $\mathbb{C}^d \otimes \mathbb{C}^d$, as
\begin{equation}
    \Omega_{\rho} \coloneqq \sum_{i,j=1}^d \rho_{ij} |ii\rangle\langle jj| . \label{eq:1}
\end{equation}

In the present work, we direct our focus to a specific family of maximally correlated states associated with states of the form
\begin{equation}
    \rho_V \coloneqq \frac{1}{2k} \begin{pmatrix} \id & V^{\vphantom{\dagger}} \\ V^\dagger & \id \end{pmatrix} , \label{eq:22}
\end{equation}
where $V$ is an arbitrary $k \times k$ unitary matrix, and the total dimension $d = 2k$. For simplicity, we denote these maximally correlated states by $\Omega_{V} \coloneqq \Omega_{\rho_V}$. 

We refer to these constructions as `generalised flower states'. This nomenclature is motivated by the fact that choosing $V$ to be the $k$-dimensional Fourier transform
\begin{equation}
F_k\coloneqq \frac{1}{\sqrt{k}} \sum_{p,q=0}^{k-1} \omega_k^{-pq} |p\rangle\langle q| , \label{eq:24}
\end{equation}

where $\omega_k\coloneqq e^{-2\pi i/k}$ is the $k$-th root of unity, recovers the flower states $\Omega_{F_k}$ introduced in~\cite{horodecki2005locking}, and further studied by Christandl and Winter~\cite{christandl2005uncertainty} and by Yang et al.~\cite{yang2009squashed}. To make the connection with these works more explicit, we could express $\Omega_{F_k}$ as
\begin{equation}
\Omega_{F_k} = \frac{1}{2k} \sum_{i,i',j,j'} \langle j'|F_k^{i-i'}|j\rangle |ij\rangle\langle i'j'|_{A_1A_2} \otimes |ij\rangle\langle i'j'|_{B_1B_2}  \label{eq:25}
\end{equation}
where $i, i' \in \{0, 1\}$ and $1 \le j, j' \le k$. Here we split the local $2k$-dimensional systems $A = A_1A_2$ and $B = B_1B_2$ in two parts each, the first of which, \textit{i.e.}\ $A_1, B_1$, are single qubits, and the second of which, \textit{i.e.}\ $A_2, B_2$,  are $k$-dimensional. 

Remarkably, the squashed entanglement, the entanglement of formation, and even the LOCC entanglement cost of the flower states can be computed exactly: all of them turn out to equal $\log \big( 2\sqrt{k} \big)$ (\cite{christandl2005uncertainty}, Section III).

\section{
Non-entangling and PPT Operations} Having recalled the structural properties of flower states, we now evaluate their distillable entanglement and entanglement cost. For a given class of free operations $\cleancal{F}$, the distillable entanglement $E_{d,\cleancal{F}}(\rho)$ denotes the maximum rate at which two-qubit maximally entangled states  can be extracted from $\rho$ with asymptotically vanishing error using operations in $\FF$. Conversely, the entanglement cost $E_{c,\cleancal{F}}(\rho)$ characterises the minimum rate of ebits required to dilute copies of $\rho$ under the same asymptotic constraints. In the main manuscript, we focus specifically on non-entangling operations, which, by definition, are required to map separable states to separable states.
In the appendix, we discuss  PPT-preserving operations as well.

It is noteworthy to point out that, while the above quantities are more practically reasonable, one might also be interested theoretically in the exact, zero-error versions of these quantities:  $E^{\mathrm{exact}}_{d,\cleancal{F}}(\rho)$ and $E^{\mathrm{exact}}_{c,\cleancal{F}}(\rho)$. Despite the absolute theoretical nature of these quantities, these exact quantities become especially useful in establishing bounds since: 

\begin{equation}
E^{\mathrm{exact}}_{d,\cleancal{F}}(\rho)
\le
E_{d,\cleancal{F}}(\rho)
\le
E_{c,\cleancal{F}}(\rho)
\le
E^{\mathrm{exact}}_{c,\cleancal{F}}(\rho)
\end{equation}

For the exact entanglement cost under NE, it is governed by the standard robustness of entanglement ($R^s_{\cleancal{S}}$)~\cite{vidal1999}. As a reminder, the  standard robustness of entanglement  is defined as:
\begin{equation}
R^s_{\cleancal{S}}(\omega)\coloneqq \inf \{\tr L : L, \omega + L \in \cleancal{S}\} . \label{eq:3}
\end{equation}

Robustness provides essential bounds, ultimately letting us determine the exact entanglement cost under NE operations~\cite{Brando2010,nilanjana2011}, whereas to evaluate the standard NE cost, we explicitly compute the recently introduced tempered negativity~\cite{Lami2023} in  \hyperref[appendixA]{Appendices A} and \hyperref[appendixB]{B}, respectively. Remarkably, this measure coincides precisely with the exact NE cost, which we explicitly demonstrate in \hyperref[appendixC]{Appendix C} leading to the following theorem: 


\begin{theorem}\label{theo2}
Let $k \ge 2$. The flower state $\Omega_{F_k}$ satisfies
\begin{equation}
E_{d, \sepp}(\Omega_{F_k}) = E_R^\infty(\Omega_{F_k}) = E_R(\Omega_{F_k}) = \log 2, \label{eq:26}
\end{equation}
but
\bb
E_{c, \locc}(\Omega_{F_k}) &= \Esq (\Omega_{F_k}) \\
&= \log\left(2\sqrt{k}\right) \\
&> \log\left(1+\sqrt{k}\right) \\
&= E_{c,\sepp}^{\text{exact}}(\Omega_{F_k}) \\
&= E_{c,\sepp}(\Omega_{F_k}) \label{eq:27}
\ee

Also,
\bb\label{eqn10}
E_{c,\ppt}^{\text{exact}}(\Omega_{F_k}) &= \log \big\|\Omega_{F_k}^\Gamma\big\|_1\\
&= \log \big(1+\sqrt{k}\big) \\
&= E_{c,\sepp}^{\text{exact}}(\Omega_{F_k})\, ,
\ee
\end{theorem}
where $X^\Gamma$ denotes the partial transpose $\Gamma(X)$ of an operator $X$ and the trace norm is given by $\| \cdot \|_1 = \tr | \cdot |$,   where $|X| \coloneqq \sqrt{X^\dagger X}$ defines its absolute value. \eqref{eqn10} follows from the fact that $|\Omega_{F_k}^\Gamma|$ is diagonal in a product basis, hence  $|\Omega_{F_k}^\Gamma|^\Gamma \ge 0$. Therefore, $\Omega_{F_k}$ has `zero binegativity'
, entailing that the logarithmic negativity~\cite{vidal2002computable,plenio2005logarithmic,Satoshi,werner02} coincides with the exact entanglement cost under PPT operations~\cite{audenaert2003entanglement,wang2020cost}. 

Interestingly, several important implications follow directly from this theorem:
\begin{enumerate}[label=\textbf{Implication \arabic*}, wide=0pt]
\item \label{item:a}: The flower states form a family of states that is strongly irreversible under all classes of operations that are a subset of NE, in the sense that for any such class the distillable entanglement is limited by a constant but the corresponding entanglement cost is unbounded.

\item \label{item:b}: Enlarging the set of operations from LOCC to NE can grant a true advantage in asymptotic entanglement dilution~\footnote{Brand\~ao and Plenio's results imply that for all bound entangled states there is always an advantage in distillation, as all entangled states would become NE-distillable.}, as~\eqref{eq:27} directly shows.

\item\label{item:c}: In general, \emph{the squashed entanglement 
is not monotonic under non-entangling operations; if it were, then we would have $\Esq \leq E_{c,\sepp}$, which is violated on flower states by~\eqref{eq:27}.} This means that it cannot be used to establish the fundamental irreversibility of entanglement theory advocated in~\cite{Lami2023}.
\end{enumerate}
\ref{item:a}  reveals a deeper physical consequence. While the qualitative breakdown of a second law for entanglement is well established~\cite{Lami2023}, our results advance this understanding by providing a strictly quantitative characterisation of this irreversibility. For this, we evaluate the asymptotic normalised non-entangling irreversibility gap, defined as:$$\Delta_{\mathrm{irr}}^{\sepp} \coloneqq \limsup_{d\to\infty} \frac{1}{\log d} \sup_{\rho\in\mathcal{D}(\mathbb{C}^d \otimes \mathbb{C}^d)} \left[ E_{c,\sepp}(\rho) - E_{d,\sepp}(\rho) \right] .$$ 
 which captures the worst-case scenario for entanglement manipulation, representing the maximum  fraction of entanglement that is fundamentally destroyed in the macroscopic limit.
 From~\cite{Lami2023}, it was established that $\Delta_{\mathrm{irr}}^{\sepp} \ge \frac{1}{\log 3} \left( 1 - \log \frac{3}{2} \right) \approx 0.262$.

In this context, Theorem~\ref{theo2} gives us an improved estimate for the lower bound, which we restate as follows.

\begin{corollary}\label{cor1}
The asymptotic irreversibility gap under non-entangling operations satisfies:
\begin{equation}
\Delta_{\mathrm{irr}}^{\sepp} \ge \frac{1}{2} . 
\end{equation}
\end{corollary}

\begin{proof}
For each fixed $k$, we have that
\begin{align}
\Delta_{\mathrm{irr}}^{\sepp}
&\ge
\frac{1}{\log(2k)}
\left(
E_{c,\sepp}(\Omega_{F_k})
-
E_{d,\sepp}(\Omega_{F_k})
\right)
\nonumber\\
&=
\frac{\log k}{2(\log 2 + \log k)} .
\label{eq:35}
\end{align}
We can then take the limit $k \to \infty$, which yields the claim and establishes that the largest known `irreversibility gap' is in fact $\Theta \left( \frac{1}{2}\log d \right)$.
\end{proof}

\section{Zero-Error Entanglement Manipulation Under LOCC}

Having explicitly explored the exact and standard entanglement manipulation under non-entangling operations, let us shift our attention to the  zero-error LOCC framework.

Under LOCC, the exact  cost is governed by the Schmidt number of the state~\cite{nielsen99}, which extends the definition of the Schmidt rank from pure states to mixed states via the convex-roof construction:
\begin{equation}
r_S(\rho) = \inf_{\{p_i, |\psi_i\rangle\}} \max_i r_S(\psi_i),
\end{equation}
where the infimum is taken over all pure-state decompositions satisfying $\rho = \sum_i p_i |\psi_i\rangle\langle\psi_i|$. Consequently, the single-copy, zero-error LOCC entanglement cost is given by:
\begin{equation}
E_{c,\locc}^{(0)}(\rho) = \log_2 r_S(\rho).
\end{equation}

However, the single-copy Schmidt number is submultiplicative~\cite{horodecki00,Yue2019}. Hence, evaluating the exact LOCC entanglement cost necessitates an asymptotic regularisation, defined as:
\bb
E_{c,\locc}^{\text{exact}}(\rho) &= \lim_{n\to\infty} \frac{1}{n} E_{c,\locc}^{(0)}(\rho^{\otimes n}) \\
&= \lim_{n\to\infty} \frac{1}{n} \log_2 \left[ r_S(\rho^{\otimes n}) \right]
\ee

Hence, it comes as no surprise that evaluating the regularised Schmidt number, and  therefore the exact LOCC cost,  remains elusive for most non-trivial states since it boils down to an intractable optimisation problem. Remarkably, we show that the flower states admit an exact analytical solution. We compute this quantity by leveraging uncertainty relations over finite abelian groups, whose full derivation we provide in  \hyperref[appendixD]{Appendix D}.

\begin{theorem}\label{theo3}
Under $\locc$, the regularised zero-error entanglement cost of the flower states is given by:
\begin{equation}
E_{c,\locc}^{\text{exact}}(\Omega_{F_k}) = \log_2 \left[ \min_{r|k} \left( r + \frac{k}{r} \right) \right].
\end{equation}
\end{theorem}

The physical and structural significance of Theorem~\ref{theo3} is remarkable, particularly when contrasted with the NE framework. It reveals that the zero-error LOCC cost is deeply tied to the number-theoretic structure of the local dimension $k$, unlike in the NE framework—where we demonstrate that the exact cost is entirely immune to the strict zero-error constraint. For example, when $k$ is a prime number, the zero-error entanglement cost reduces to $\log(k+1)$, representing an approximate twofold increase over its standard, vanishing-error LOCC counterpart. This substantial difference demonstrates that under the constraint of LOCC, preparing states with strict zero-error requires a significant amount of additional entanglement.

While the high zero-error entanglement cost shows that preparing these states is highly resource-intensive, extracting entanglement from them tells a vastly different story.


\begin{theorem}\label{theo4} The single ebit of entanglement distillable from generalised flower states under $\mathrm{NE}$ or $\mathrm{PPT}$ operations can be achieved deterministically via a one-shot, one-way $\mathrm{LOCC}$ protocol.\end{theorem}

We prove this  in \hyperref[appendixF]{Appendix E}. The above theorem carries 
significant physical implications. Typically, achieving the theoretical limit of entanglement manipulation under a class of operations  requires an asymptotic number of copies. Theorem~\ref{theo4} demonstrates that for the distillation of the generalised flower states using NE or PPT operations, this asymptotic regime is completely unnecessary as there exists a one-shot, one-way LOCC protocol which is just as powerful as NE and PPT.

It is noteworthy to point out that flower states with prime $k$ exhibit the asymptotically largest gap between exact LOCC dilution and distillation: the cost is $\log (1+d/2) \approx \log d$, while the distillable entanglement is $1$.  So $\Delta_{\mathrm{irr}}^{\text{exact, LOCC}} = 1$. On the contrary, when $k$ is a perfect square, the exact LOCC cost coincides with its standard counterpart, and both equal $\log\big(2\sqrt{k}\big)$.

\section{Discussions} In this work, we 
have explored the entanglement landscape by leveraging the unique structure of flower states. Contrary to the intuition that non-entangling (NE) operations should arbitrarily enhance dilution and distillation rates, we have demonstrated that this advantage is strictly bounded: the NE entanglement cost is at most one ebit less than its LOCC counterpart. Furthermore, the distillable entanglement is strictly limited to a single ebit. Remarkably, this limit is achievable deterministically via a one-shot, one-way LOCC protocol, rendering both asymptotic resource manipulation and the broader class of NE or PPT operations entirely redundant for this task.

By analysing the entanglement properties of these states, we have established an improved asymptotic irreversibility gap under NE operations. We have also demonstrated that the squashed entanglement fails to remain a monotone under this broader class, indicating that established LOCC entanglement monotones need not inherently retain their operational validity when the set of free operations is relaxed to NE.

Investigating the strict regime of exact, zero-error LOCC transformations, we have derived the exact analytical value of the coherence number for $\rho_{F_k}$ (and by extension, the Schmidt number of $\Omega_{F_k}$), which shows that flower states with prime $k$ exhibit the asymptotically largest gap between exact LOCC dilution and distillation: the cost is $\log (1+d/2) \approx \log d$, while the distillable entanglement is $1$.

These findings open several compelling avenues for future research. First, while we have derived the lower bound $\Delta_{\mathrm{irr}}^{\sepp} \ge \frac{1}{2}$, it remains an open question whether specific quantum states can saturate the theoretical maximum of $\Delta_{\mathrm{irr}}^{\sepp} = 1$. Identifying such states would reveal the ultimate ``resource traps'' in entanglement theory. Second, determining which other established measures fail to preserve monotonicity under NE operations will systematically map the operational boundary between strictly physical local operations and their broader axiomatic relaxations.

Finally, the fact that generalised flower states can be distilled to their asymptotic threshold using only a simple, one-shot deterministic one-way LOCC protocol is of independent interest. Identifying other classes of states that share this property could reveal a deeper structural reason for why the theoretical power of NE operations is sometimes practically redundant.

\section{acknowledgements} We would like to thank Bartosz Regula for several helpful discussions. We acknowledge financial support from the European Union (ERC StG ETQO, Grant Agreement no.\ 101165230). Views and opinions expressed are however those of the author(s) only and do not necessarily reflect those of the European Union or the European Research Council. Neither the European Union nor the granting authority can be held responsible for them. 

OpenAI's ChatGPT~5.5 was employed to assist us in searching the literature on uncertainty relations in cyclic groups, and in completing the proof of Theorem~\ref{theo3}.

\appendix
\bibliographystyle{apsc}
\bibliography{bibliography}
\clearpage
    \pagebreak

\resumetoc

    \setcounter{equation}{0}
    \setcounter{figure}{0}
    \setcounter{table}{0}
    \setcounter{section}{0}

    \renewcommand{\theequation}{S\arabic{equation}}
    \renewcommand{\thefigure}{S\arabic{figure}}
    \renewcommand{\thetable}{S\arabic{table}}
    \renewcommand{\thesection}{S\roman{section}}

    \onecolumngrid 
    \begin{center}
        \textbf{\large Supplemental Material: \\[.2em]
        Entanglement of flower states}\\[.8em]
        Samrat Sen$^1$ and  Ludovico Lami$^1$\\[.2em]
        \small{\textit{$^1$Scuola Normale Superiore, Piazza dei Cavalieri 7, 56126 Pisa, Italy}}\\
      
    \end{center}
    \addtocontents{toc}{\protect\setcounter{tocdepth}{2}}
    \vspace{10pt}

    \tableofcontents
    \vspace{20pt}

    \section{Prior Work}

Before moving on to the main proofs of the paper, one thing needs to be understood very well: the shift towards an axiomatic approach in entanglement theory is not just a convenient way to bypass the mathematical complexities of LOCC. It is actually driven by a deep connection to thermodynamics. Classical thermodynamics is itself an axiomatic theory, and it shares many structural similarities with entanglement. However, a major difference arises when it comes to reversibility. Under LOCC, pure states are perfectly reversible, mirroring the elegance of ideal thermodynamic processes. In contrast, mixed entangled states are stubbornly irreversible: the entanglement cost required to create them strictly exceeds the amount one can distil from them ($E_c > E_d$). The desire to restore this reversibility—and to find a unique ``second law'' governing all entanglement transformations~\cite{Popescu1997Thermodynamics,Vedral1998Entanglement,Vidal2000Entanglement,Horodecki2002Thermodynamical}—naturally inspired the field to adopt an axiomatic framework, building the theory directly from the core physical principles of entanglement manipulation.

Asymptotically non-entangling (ANE) operations appeared to restore reversibility by identifying the regularised relative entropy of entanglement ($E_R^\infty$) as the universal quantifier~\cite{Brando2008,Brando2010,Brando_stein}:

\bb
E_d(\rho_{AB}) &= E_c(\rho_{AB}) \\
&= E_R^\infty(\rho_{AB}) \\
&= \lim_{n \to \infty} \frac{1}{n} \min_{\sigma_{A_n B_n} \in ~\sep} D(\rho_{AB}^{\otimes n} \| \sigma_{A_n B_n})
\ee

Although a technical gap regarding the generalised quantum Stein's lemma~\cite{berta} temporarily cast doubt on the reversible theory, a recent proof by one of the authors has re-established its validity under ANE~\cite{lami_Stein}. Nevertheless, the thermodynamic parallel is definitively shattered by the presence of irreversibility under NE, the largest resource non-generating class that preserves separability~\cite{Lami2023}.

Let us recall that the two-qutrit state $\omega_3$ introduced in~\cite{Lami2023}, which is a counterexample to the existence of  a `second law of entanglement theory',   satisfies $E_{c,\sepp}(\omega_3) = 1$ and $E_{d,\sepp}(\omega_3) = \log(3/2)$.
 Hence, for the  state $\rho = \omega_3^{\otimes n}$, the total local dimension scales as $d = 3^n$. As any regularised quantity, such as $E_c$ and $E_d$ , is additive, we have $E(\omega_3^{\otimes n}) = n E(\omega_3)$ and thus 
\bb
\Delta_{\mathrm{irr}}^{\sepp} &\ge \lim_{n\to\infty} \frac{n \left[ E_{c,\sepp}(\omega_3) - E_{d,\sepp}(\omega_3) \right]}{n \log 3} \\
&= \frac{E_{c,\sepp}(\omega_3) - E_{d,\sepp}(\omega_3)}{\log 3} \\
&= \frac{1}{\log 3} \left( 1 - \log \frac{3}{2} \right) \\
&\approx 0.262 .
\ee establishing the previously known lower bound.\\

Since we shall be dealing with  standard as well as exact entanglement manipulation, it will be helpful to recall the  definitions of distillable entanglement and entanglement cost for the corresponding regimes. Following Vidal and Werner [\cite{vidal2002computable}, Eq. (43)], for an arbitrary $\varepsilon \in [0, 1)$, distillable entanglement and entanglement cost of a state $\rho_{AB}$ under $\cleancal{K}$-preserving operations, where $\cleancal{K}=\cleancal{S}, \cleancal{PPT}$ are given as:  
\bb
E^\varepsilon_{d,\mathrm{KP}}(\rho_{AB}):= \sup \left\{ R>0 \;\middle|\; \limsup_{n\to\infty} \inf_{\Lambda_n \in \mathrm{KP}\left(A^nB^n \to A_0^{\lceil Rn\rceil} B_0^{\lceil Rn\rceil}\right)} \frac12 \left\| \Lambda_n\left(\rho_{AB}^{\otimes n}\right) - \Phi_2^{\otimes \lceil Rn\rceil} \right\|_1 \le \varepsilon \right\} .
\ee

\bb
E^\varepsilon_{c,\mathrm{KP}}(\rho_{AB}):= \inf \left\{ R>0 \;\middle|\; \limsup_{n\to\infty} \inf_{\Lambda_n \in \mathrm{KP}\left(A_0^{\lfloor Rn\rfloor} B_0^{\lfloor Rn\rfloor} \to A^nB^n\right)} \frac12 \left\| \Lambda_n\left(\Phi_2^{\otimes \lfloor Rn\rfloor}\right) - \rho_{AB}^{\otimes n} \right\|_1 \le \varepsilon \right\} .
\ee
In contrast to the standard entanglement cost, which quantifies asymptotic dilution by allowing an error that vanishes in the many-copy limit, the zero-error regime governs exact transformations. The corresponding  entanglement measures
read:

\bb
E^{\mathrm{exact}}_{d,\mathrm{KP}}(\rho_{AB}):= \sup \;\Biggl\{ R>0 \;\Bigm|\; \exists\, n_0 \ \text{s.t.}\ \forall n \ge n_0,\ \exists\, \Lambda_n \in \mathrm{KP}\left(A^nB^n \to A_0^{\lceil Rn\rceil}B_0^{\lceil Rn\rceil}\right), \Lambda_n(\rho_{AB}^{\otimes n}) = \Phi_2^{\otimes \lceil Rn\rceil} \Biggr\} ,
\ee

\bb
E^{\mathrm{exact}}_{c,\mathrm{KP}}(\rho_{AB})
:= \inf \;\Biggl\{ R>0 \;\Bigm|\; &\exists\, n_0 \ \text{s.t.}\ \forall n \ge n_0,\ \exists\, \Lambda_n \in \mathrm{KP}\left(A_0^{\lfloor Rn\rfloor}B_0^{\lfloor Rn\rfloor} \to A^nB^n\right), 
&\Lambda_n(\Phi_2^{\otimes \lfloor Rn\rfloor}) = \rho_{AB}^{\otimes n} \Biggr\} ,
\ee
Hence, it trivially follows that: 
\bb
E^{\mathrm{exact}}_{d,\mathrm{KP}}(\rho)
\le
E_{d,\mathrm{KP}}(\rho)
\le
E_{c,\mathrm{KP}}(\rho)
\le
E^{\mathrm{exact}}_{c,\mathrm{KP}}(\rho) .
\ee

\section{Appendix A: Standard robustness of entanglement and exact NE entanglement cost}\label{appendixA}
We shall now be concerned with the problem of estimating the standard robustness of entanglement~\cite{vidal1999,lami2021framework} as well as the closely related robustness of PPT-ness 
 of the flower states. 
We remind the reader that for $\cleancal{K} = \cleancal{S}, \cleancal{P}\cleancal{P}\cleancal{T}$ 

\bb
R^s_{\cleancal{K}}(\omega)\coloneqq \inf \{\tr L : L, \omega + L \in \cleancal{K}\} . \label{eq:15}
\ee

However, before that, we need to define the following operator: 

\bb
\Omega[S, T]\mathrel{\mathop:}= \sum_{i,j=1}^d S_{ij} |ii\rangle\langle jj| + \sum_{i \neq j} T_{ij} |ij\rangle\langle ij| . 
\ee
which acts on a bipartite system with Hilbert space $\mathbb{C}^d \otimes \mathbb{C}^d$, given two $d \times d$ matrices $S, T$, with $S$ Hermitian and $T$ real.

In what follows, we will often extend $T$ to have zero diagonal, which we write $\Delta(T) = 0$. These operators have been considered by Johnston and MacLean~\cite{johnston2021pairwise}, and in even greater generality by Singh and Nechita~\cite{singh2021diagonal}.

We will also use the shorthand notation
\begin{equation}
\Omega_S\mathrel{\mathop:}= \Omega[S, 0] . \label{eq:2}
\end{equation}
Hence, in the case where $S$ is a state, \textit{i.e.}\ a density matrix, it becomes  $\Omega_S$, the maximally correlated state associated with $S$. The following lemma is crucial in calculating the exact NE cost of the flower states as we shall see later:

\begin{Lemma}\label{l1}
Let $\rho$ be a $d \times d$ density matrix. Then the robustnesses of the associated maximally correlated state $\Omega_\rho$ satisfy that
\begin{equation}
\frac{\|\rho\|_{\ell_1} - 1}{2} \le R^s_{\cleancal{PPT}}(\Omega_\rho) \le R^s_{\cleancal{S}}(\Omega_\rho) \le \|\rho\|_{\ell_1} - 1  \label{eq:4}
,\end{equation} where $\|\cdot\|_{\ell_1}$ denotes the entry-wise $\ell_1$ norm (also referred to as the Hadamard $1$-norm or taxicab matrix norm), defined for an arbitrary $d \times d$ matrix $A$ as the absolute sum of all its entries:$$\|A\|_{\ell_1} \coloneqq \sum_{i=1}^d \sum_{j=1}^d |A_{ij}|$$
\end{Lemma}

\begin{proof}
The lower bounds
\begin{equation}
R^s_{\cleancal{S}}(\Omega_\rho) \ge R^s_{\cleancal{PPT}}(\Omega_\rho) \ge \frac{\|\Omega_\rho^\Gamma\|_1 - 1}{2} = \frac{\|\rho\|_{\ell_1} - 1}{2} \label{eq:5}
\end{equation}
are known from general arguments. The problem is to prove the upper bounds.

Consider a pair of indices $i < j$ such that $\rho_{ij} \neq 0$. Setting $\varphi_{ij}\coloneqq \frac{\rho_{ij}}{|\rho_{ij}|}$, define the state
\bb
\omega_{ij}^\pm
&\coloneqq
\frac{1}{2}
\left( \frac{\varphi_{ij}|ii\rangle \pm |jj\rangle}{\sqrt{2}} \right)
\left( \frac{\varphi_{ij}^*\langle ii| \pm \langle jj|}{\sqrt{2}} \right)
+ \frac{1}{4} \big( |ij\rangle\langle ij| + |ji\rangle\langle ji| \big) \\
&=
\frac{1}{4|\rho_{ij}|}
\Omega\Big[
|\rho_{ij}| (|i\rangle\langle i| + |j\rangle\langle j|) \pm \rho_{ij}|i\rangle\langle j| \pm \rho_{ij}^*|j\rangle\langle i| , \,\,
|\rho_{ij}| (|i\rangle\langle j| + |j\rangle\langle i|)
\Big].
\ee

It is very easy to verify that $\omega_{ij}^\pm$ is a PPT state; since it lives on an effective two-qubit system with local Hilbert spaces $\operatorname{span}\{|i\rangle , |j\rangle\}$, it must be separable, too~\footnote{This is also easy to see directly, but let’s not get into that.}. Therefore, denoting with $|M|$ the entry-wise absolute value of a matrix $M$, let us write
\bb
\Omega_\rho + 2 \sum_{i<j:\rho_{ij}\neq0} |\rho_{ij}|\,\omega_{ij}^-
&= \Omega\Bigg[ \rho + \frac{1}{2} \sum_{i<j} \Big( |\rho_{ij}| (\ketbra{i} + \ketbra{j}) - \rho_{ij}\ketbraa{i}{j} - \rho_{ij}^*\ketbraa{j}{i} \Big) , \,\, \frac{1}{2} \sum_{i<j} |\rho_{ij}| (\ketbraa{i}{j} + \ketbraa{j}{i}) \Bigg] \\
&= \Omega\Bigg[ \frac{1}{2}\rho + \frac{1}{2} \sum_i \Big( \sum_j |\rho_{ij}| \Big) \ketbra{i} , \quad \frac{1}{2}\big(|\rho|-\Delta(\rho)\big) \Bigg] \\
&= \sum_i \rho_{ii} \ketbra{ii} + 2 \sum_{i<j: \rho_{ij} \neq 0} |\rho_{ij}| \omega^+_{ij} .
\ee
 Since the operator on the last line is a sum of separable states and hence it is clearly separable itself, we deduce that
\begin{equation}
R^s_{\cleancal{S}}(\Omega_\rho) \le 2 \sum_{i<j: \rho_{ij} \neq 0} |\rho_{ij}| = \|\rho\|_{\ell_1} - 1 , \label{eq:8}
\end{equation}

which establishes (\ref{eq:4}).
\end{proof}

\begin{rem}
A sloppier estimate that is anyway sufficient for our purposes rests on the observation that $\Omega_\rho + 4 \sum_{i<j: \rho_{ij} \neq 0} |\rho_{ij}| \omega_{ij}^-$ is manifestly separable because it is diagonal in a product basis. This yields $R^s_{\cleancal{S}}(\Omega_\rho) \le 2 (\|\rho\|_{\ell_1} - 1)$, which — despite being sub-optimal — would suffice to establish the forthcoming Corollary~\ref{cor2}.
\end{rem}

\begin{corollary}\label{cor2}
For every $d \times d$ density matrix $\rho$ and $\cleancal{K} = \cleancal{S}, \cleancal{PPT}$, the associated maximally correlated state $\Omega_\rho$ satisfies that
\begin{equation}
E_{c, \cleancal{K}}^{\text{exact}}(\Omega_\rho) = \log \|\rho\|_{\ell_1} . \label{eq:9}
\end{equation}
\end{corollary}

\begin{proof}
On the one hand, using the upper bound in (\ref{eq:4}) we have that
\begin{align}
E_{c, \cleancal{K}}^{\text{exact}}(\Omega_\rho) &= \lim_{n \to \infty} \frac{1}{n} \log \left( 1 + 2 R^s_{\cleancal{K}} \left( \Omega_\rho^{\otimes n} \right) \right) \nonumber \\
&\le \lim_{n \to \infty} \frac{1}{n} \log \left( 2 \|\rho^{\otimes n}\|_{\ell_1} - 1 \right) \nonumber \\
&= \lim_{n \to \infty} \frac{1}{n} \log \left( 2 \|\rho\|_{\ell_1}^n - 1 \right) \nonumber \\
&= \log \|\rho\|_{\ell_1} . \label{eq:10}
\end{align}
On the other, leveraging the lower bound in (\ref{eq:4}) we have that
\begin{equation}
E_{c, \cleancal{K}}^{\text{exact}}(\Omega_\rho) = \lim_{n \to \infty} \frac{1}{n} \log \left( 1 + 2 R^s_{\cleancal{K}} \left( \Omega_\rho^{\otimes n} \right) \right) \ge \log \|\rho\|_{\ell_1} , \label{eq:11}
\end{equation}
which completes the proof.
\end{proof}

\begin{rem}
For the exact $\mathrm{PPT}$ entanglement cost, Audenaert, Plenio, and Eisert established in their seminal work~\cite{audenaert2003entanglement} that $E_{c,\ppt}^{\text{exact}}(\rho)$, can be computed analytically for any bipartite state $\rho = \rho_{AB}$ that exhibits `zero binegativity'~\cite{audenaert2003entanglement,Satoshi,werner02}. This condition is formally expressed as $|\rho^\Gamma|^\Gamma \ge 0$, where $X^\Gamma$ denotes the partial transpose $\Gamma(X)$ of an operator $X$, and $|X| \coloneqq \sqrt{X^\dagger X}$ defines its absolute value. Whenever this criterion is met, the exact PPT cost simplifies to the logarithmic negativity~\cite{vidal2002computable,plenio2005logarithmic,wang2020cost}, obtained by evaluating the trace norm $\| \cdot \|_1 \coloneqq \text{Tr}|\cdot|$ of the partially transposed state:
\begin{equation}
    E_{c,\ppt}^{\text{exact}}(\rho) = \log \|\rho^\Gamma\|_1 \coloneqq E_N(\rho) .
\end{equation}

Since $|\Omega_\rho^\Gamma|$ is diagonal in a product basis and hence $|\Omega_\rho^\Gamma|^\Gamma \ge 0$, \textit{i.e.}\ $\Omega_\rho$ has zero binegativity, 
\begin{equation}
E_{c, \ppt}^{\text{exact}}(\Omega_\rho) = \log \|\rho\|_{\ell_1} . \label{eq:12}
\end{equation}
We could have used this information to obtain a more direct proof of the inequality $E_{c, \text{PPTP}}^{\text{exact}}(\Omega_\rho) \le \log \|\rho\|_{\ell_1}$, as in $E_{c, \text{PPTP}}^{\text{exact}}(\Omega_\rho) \le E_{c, \ppt}^{\text{exact}}(\Omega_\rho) = \log \|\Omega_\rho^\Gamma\|_1 = \log \|\rho\|_{\ell_1}$.
\end{rem}

\section{Appendix B: Tempered negativity and standard entanglement cost under NE}\label{appendixB}

Let us remember 
that the tempered negativity~\cite{Lami2023} of a state $\omega$ is defined as 
\begin{equation}
N_\tau(\omega)\coloneqq \sup \left\{ \tr \omega X: \|X^\Gamma\|_\infty \le 1, \tr \omega X = \|X\|_\infty \right\} . \label{eq:13}
\end{equation}
Its logarithmic version can be used to lower bound the entanglement cost under $\cleancal{K}$-preserving operations (here as usual $\cleancal{K} = \cleancal{S}, \cleancal{PPT}$) at arbitrary error threshold $\varepsilon \in [0, 1/2)$, as in~\cite{Lami2023}
\begin{equation}
E_{c, \cleancal{K}}^{\text{exact}}(\omega) \ge E_{c, \cleancal{K}}^\varepsilon(\omega) \ge E_\tau^N(\omega)\coloneqq \log N_\tau(\omega) . \label{eq:14}
\end{equation}
It is therefore of interest to try to compute $N_\tau$ in some cases of interest. Here we provide a new result in this direction:

\begin{proposition}\label{prop1}
Let $\rho = \frac{\Pi_r}{r}$ be a normalised projector of rank $r < d$ inside a $d$-dimensional space. Then, defining
\begin{equation}
\lambda\coloneqq \max_i \rho_{ii} , \quad \mu\coloneqq \min_i \rho_{ii} , \quad \gamma\coloneqq \max_{i \neq j} |\rho_{ij}| , \label{eq:15}
\end{equation}
we have that the tempered negativity of the associated maximally correlated state $\Omega_\rho$ is at least

\bb
N_\tau(\Omega_\rho) \ge \max_{a,b \in \mathbb{R}} \Big\{ a + \frac{b}{r} \;\Big|\; \max \big\{|a + \lambda b| , |a + \mu b| , b\gamma\big\} \le 1, \ a + \frac{b}{r} \ge |a| \Big\} . \label{eq:16}
\ee

In particular, if $\lambda = \mu = 1/d$, \textit{i.e.}\ $\rho$ has uniform diagonal, and $r \le d/2$, then
\begin{equation}
N_\tau(\Omega_\rho) \ge 1 + \frac{1}{\gamma} \left( \frac{1}{r} - \frac{1}{d} \right) . \label{eq:17}
\end{equation}
\end{proposition}

\begin{proof}
Constructing the ansatz
\begin{equation}
X = a\mathbb{P} + b \Omega_\rho , \label{eq:18}
\end{equation}
where $\mathbb{P}$ is the projector onto the maximally correlated subspace. Then
\bb \label{eq:19}
\|X\|_\infty &= \|a\id + b\rho\|_\infty \\ &= \max \left\{ \left| a + \frac{b}{r} \right| , |a| \right\} ,
\ee
while
\bb \label{eq:20}
\begin{aligned}
\tr X\Omega_\rho &= a + b \tr \rho^2 \\ 
&= a + \frac{b}{r} .
\end{aligned}
\ee
Therefore, the condition $\|X\|_\infty = \tr X\Omega_\rho$ is met iff $a + \frac{b}{r} \ge |a|$. Now, since it is easy to see that
\begin{equation}
\|X^\Gamma\|_\infty = \max \{|a + \lambda b| , |a + \mu b| , b\gamma\} , \label{eq:21}
\end{equation}
the last condition needed for $X$ to be a valid ansatz in the tempered negativity, \textit{i.e.}\ $\|X^\Gamma\|_\infty \le 1$, becomes the one stated in (\ref{eq:16}).

If $\lambda = \mu = 1/d$ and $r \le d/2$ we can make the further ansatz $a = 1 - \frac{1}{d\gamma}$ and $b = \frac{1}{\gamma}$, which yields precisely (\ref{eq:17}).
\end{proof}

\begin{corollary}\label{cor3}
Let $V$ be a $k \times k$ unitary matrix. Then the maximally correlated state $\Omega_V$ associated with (\ref{eq:22}) satisfies that
\begin{equation}
N_\tau(\Omega_V) \ge 1 + \frac{1}{\delta(V)} , \label{eq:23}
\end{equation}
where $\delta(V)\coloneqq \max_{i,j} |V_{ij}|$.
\end{corollary}

\begin{proof}
We start by noting that since $(2k\rho_V - \id)^2 = \id$, the operator $2k\rho_V - \id$ can only have eigenvalues $\pm 1$. Given that its trace is $0$, it has in fact $k$ eigenvalues equal to $+1$ and $k$ equal to $-1$. It follows immediately that $k\rho_V$ is in fact a projector of rank $k = d/2$. Since it has uniform diagonal, we can apply Proposition~\ref{prop1}, and in particular the estimate (\ref{eq:17}). Setting $r = k = d/2$ and $\gamma = \frac{\delta(V)}{2k}$ yields precisely (\ref{eq:23}).
\end{proof}

\section{Appendix C: Proof of Theorem 1 }\label{appendixC}
For ease of reference, we recall Theorem~\ref{theo2}:

\begin{manualthm}{\ref{theo2}}
Let $k \ge 2$. The flower state $\Omega_{F_k}$ satisfies
\begin{equation} \tag{\ref{eq:26}}
E_{d, \sepp}(\Omega_{F_k}) = E_R^\infty(\Omega_{F_k}) = E_R(\Omega_{F_k}) = \log 2, 
\end{equation}
but
\begin{equation} \tag{\ref{eq:27}}
E_{c, \locc}(\Omega_{F_k}) = \Esq (\Omega_{F_k}) = \log\left(2\sqrt{k}\right) > \log\left(1+\sqrt{k}\right) = E_{c,\sepp}^{\text{exact}}(\Omega_{F_k}) = E_{c,\sepp}(\Omega_{F_k}) 
\end{equation}
Also,
\begin{equation} \tag{\ref{eqn10}}
E_{c,\ppt}^{\text{exact}}(\Omega_{F_k}) = \log \|\Omega_{F_k}^\Gamma\|_1 = \log (1+\sqrt{k}) = E_{c,\sepp}^{\text{exact}}(\Omega_{F_k}) .
\end{equation}
\end{manualthm}
where $X^\Gamma$ denotes the partial transpose $\Gamma(X)$ of an operator $X$, and the trace norm is given by $\| \cdot \|_1 = \tr | \cdot |$,   where $|X| \coloneqq \sqrt{X^\dagger X}$ defines its absolute value.

\begin{proof}
From the proof of Corollary~\ref{cor3}, we remember that $\Omega_{F_k}$ is a normalised projector of rank $k$. Since it is also maximally correlated, it follows from~\cite{rains01}, Theorem 6.2 that the PPT distillable entanglement of a maximally correlated state is given by the same formula as the hashing bound of~\cite{hashing}, and thus coincides with the one-way LOCC distillable entanglement. Since in general$$E_{d,\onelocc}(\omega) \leq E_d(\omega) \leq E_{d,\sep}(\omega) \leq E_{d,\ppt}(\omega),$$

for maximally correlated states all four quantities coincide and we also know that~\cite{berta}
$$E_{d,{\ane}}(\rho) = E_{d,\sepp}(\rho)=E_R^\infty(\rho_{AB})$$

Hence, we have that
\bb
E_{d,\kp}(\Omega_{F_k}) &= E_{d,\locc}(\Omega_{F_k}) \\
&= E_R^\infty(\Omega_{F_k})\\ &= E_R(\Omega_{F_k}) \\
&= S(\tr_1 \Omega_{F_k}) - S(\Omega_{F_k}) \\
&= \log(2k) - \log k \\
&= \log 2 . \label{eq:28}
\ee
We can also safely apply Corollary~\ref{cor2}, the lower bounds in (\ref{eq:14}), and Corollary~\ref{cor3}, which together yield
\bb
\log\!\left(1+\sqrt{k}\right) 
= \log\|\rho_{F_k}\|_{\ell_1} 
\overset{\text{Cor.\,\ref{cor2}}}{=} E_{c,\kp}^{\text{exact}}(\Omega_{F_k}) 
\overset{(\ref{eq:14})}{\ge} E_{c,\kp}^\varepsilon(\Omega_{F_k}) 
\overset{(\ref{eq:14})}{\ge} E_\tau^N(\Omega_{F_k}) 
\overset{\text{Cor.\,\ref{cor3}}}{\ge} \log\!\left(1+\sqrt{k}\right) . 
\label{eq:29}
\ee
for all $\varepsilon \in [0, 1/2)$. Therefore,
\begin{equation}
E_{c, \kp}^{\text{exact}}(\Omega_{F_k}) = E_{c, \kp}^\varepsilon(\Omega_{F_k}) = \log \left( 1 + \sqrt{k} \right) \label{eq:30}
\end{equation}
for all $\varepsilon \in [0, 1/2)$. Since from (\cite{christandl2005uncertainty}, Section III) we know that
\begin{equation}
E_{c, \locc}(\Omega_{F_k}) = \Esq (\Omega_{F_k}) = \log \left( 2\sqrt{k} \right) , \label{eq:31}
\end{equation}
the proof is complete.
\end{proof}

\section{Appendix D:  Regularised Schmidt number of the flower state and zero-error LOCC cost}\label{appendixD}

Since the flower states are maximally correlated states, computing their regularised Schmidt number is equivalent to computing the regularised coherence number of the state to which it is associated:
\begin{equation}
    \rho_{F_k} \coloneqq \frac{1}{2k} \begin{pmatrix} \id & F^{\vphantom{\dagger}}_{k} \\ F_k^\dagger & \id \end{pmatrix} \label{eq:22}
\end{equation}
Before going into the technical part of the proof, it is worthwhile to remember that the coherence number ($\mathrm{CN}$) of a mixed state is:
\begin{equation}
    \mathrm{CN}(\rho) = \min_{\{p_i, |\phi_i\rangle\}} \left( \max_i \text{CR}(|\phi_i\rangle) \right)
\end{equation}

With this in mind,  we define a channel  $\cleancal{E}(\rho) = \frac{1}{k^2} \sum_{m,n} (V_m W_n) \rho (V_m W_n)^\dagger$, which we shall show gives a valid convex decomposition of the state $\rho_{F_k}$. 

The controlled unitary operators $V_m$ and $W_n$ are defined as follows:
\bb
V_m &= |0\rangle\langle0| \otimes X^m + |1\rangle\langle1| \otimes Z^m \\
W_n &= |0\rangle\langle0| \otimes Z^n + |1\rangle\langle1| \otimes X^{-n}
\ee

Since the control basis states are orthogonal, we have:
\begin{equation}
    V_m W_n = |0\rangle\langle0| \otimes (X^m Z^n) + |1\rangle\langle1| \otimes (Z^m X^{-n})
\end{equation}

It is trivial to see that:
\bb
F_k^\dagger X F_k^{\vphantom{\dagger}} = Z \quad \text{and} \quad F_k^\dagger Z F_k^{\vphantom{\dagger}} = X^{-1}
\ee
Using these relations, we rewrite the target operator acting on the $|1\rangle\langle1|$ block:
\bb
\begin{aligned}
Z^m X^{-n} &= (F_k^\dagger X F_k^{\vphantom{\dagger}})^m (F_k^\dagger Z F_k^{\vphantom{\dagger}})^n \\
&= F_k^\dagger (X^m Z^n) F_k^{\vphantom{\dagger}}
\end{aligned}
\ee
 The joint twirling operator $V_m W_n$ expressed in $2 \times 2$ block diagonal matrix form is:\bb
V_m W_n &= |0\rangle\langle0| \otimes (X^m Z^n) + |1\rangle\langle1| \otimes F_k^\dagger (X^m Z^n) F_k \\
&= \begin{bmatrix} X^m Z^n & 0 \\ 0 & F_k^\dagger (X^m Z^n )F_k \end{bmatrix}
\ee

A key property of the Heisenberg-Weyl group on a $k$-dimensional space is that the set of operators $\{X^m Z^n\}_{m,n=0}^{k-1}$ forms a quantum 1-design. By Schur's Lemma, uniformly averaging any operator over this group completely depolarises the state into the identity matrix:
\bb
\frac{1}{k^2} \sum_{m=0}^{k-1} \sum_{n=0}^{k-1} (X^m Z^n) M (X^m Z^n)^\dagger = \frac{\text{Tr}(M)}{k} \cdot \id_k
\ee
An arbitrary pure state $|\psi\rangle$ within the support of $\rho_{F_k}$ has the following structure:\bb
| \psi \rangle &= \frac{1}{\sqrt{2}} \left( |0\rangle \otimes |\tilde{\psi}\rangle + |1\rangle \otimes F_k^\dagger |\tilde{\psi}\rangle \right) \\
&= \frac{1}{\sqrt{2}} \begin{bmatrix} |\tilde{\psi}\rangle \\ F_k^\dagger |\tilde{\psi}\rangle \end{bmatrix}
\ee where $|\tilde{\psi}\rangle$ is a normalised state vector in the target space.

Thus, \bb
V_m W_n |\psi\rangle &= \frac{1}{\sqrt{2}} \begin{bmatrix} X^m Z^n & 0 \\ 0 & F_k^\dagger X^m Z^n F_k \end{bmatrix} \begin{bmatrix} |\tilde{\psi}\rangle \\ F_k^\dagger |\tilde{\psi}\rangle \end{bmatrix} \\
&= \frac{1}{\sqrt{2}} \begin{bmatrix} X^m Z^n |\tilde{\psi}\rangle \\ \left( F_k^\dagger X^m Z^n F_k \right) \left( F_k^\dagger |\tilde{\psi}\rangle \right) \end{bmatrix}
\ee

Utilising the unitarity of the Fourier transform gives us: $$F_k^\dagger X^m Z^n \left( F_k F_k^\dagger \right) |\tilde{\psi}\rangle = F_k^\dagger X^m Z^n |\tilde{\psi}\rangle$$We now define the shifted target state vector as $|\tilde{\psi}_{m,n}\rangle \equiv X^m Z^n |\tilde{\psi}\rangle$. Substituting this definition back into the global state vector gives:
\bb
V_m W_n |\psi\rangle &= |\psi_{m,n}\rangle \\
&= \frac{1}{\sqrt{2}} \begin{bmatrix} |\tilde{\psi}_{m,n}\rangle \\ F_k^\dagger |\tilde{\psi}_{m,n}\rangle \end{bmatrix}
\ee
Thus,

\bb
|\psi_{m,n}\rangle\langle\psi_{m,n}| &= \frac{1}{2} \begin{bmatrix} 
|\tilde{\psi}_{m,n}\rangle\langle\tilde{\psi}_{m,n}| & |\tilde{\psi}_{m,n}\rangle\langle\tilde{\psi}_{m,n}| F_k \\ 
F_k^\dagger |\tilde{\psi}_{m,n}\rangle\langle\tilde{\psi}_{m,n}| & F_k^\dagger |\tilde{\psi}_{m,n}\rangle\langle\tilde{\psi}_{m,n}| F_k 
\end{bmatrix}
\ee

The action of the global twirling channel $\cleancal{E}(|\psi\rangle\langle\psi|)$ averages this ensemble uniformly over the indices $m, n$:\bb
\cleancal{E}(|\psi\rangle\langle\psi|) &= \frac{1}{k^2} \sum_{m,n=0}^{k-1} V_m W_n |\psi\rangle\langle\psi| (V_m W_n)^\dagger \\
&= \frac{1}{k^2} \sum_{m,n=0}^{k-1} |\psi_{m,n}\rangle\langle\psi_{m,n}|\\
&= \frac{1}{2} \begin{bmatrix} 
\left( \frac{1}{k^2} \sum_{m,n} |\tilde{\psi}_{m,n}\rangle\langle\tilde{\psi}_{m,n}| \right) & 
\left( \frac{1}{k^2} \sum_{m,n} |\tilde{\psi}_{m,n}\rangle\langle\tilde{\psi}_{m,n}| \right) F_k \\ 
F_k^\dagger \left( \frac{1}{k^2} \sum_{m,n} |\tilde{\psi}_{m,n}\rangle\langle\tilde{\psi}_{m,n}| \right) & 
F_k^\dagger \left( \frac{1}{k^2} \sum_{m,n} |\tilde{\psi}_{m,n}\rangle\langle\tilde{\psi}_{m,n}| \right) F_k 
\end{bmatrix}
\ee  Since the Heisenberg-Weyl operators $\{X^m Z^n\}$ constitute a perfect quantum 1-design, uniformly averaging over the target state ensemble completely depolarises the core space regardless of its initial structure:\bb
\frac{1}{k^2} \sum_{m,n=0}^{k-1} |\tilde{\psi}_{m,n}\rangle\langle\tilde{\psi}_{m,n}| &= \frac{1}{k^2} \sum_{m,n=0}^{k-1} (X^m Z^n) |\tilde{\psi}\rangle\langle\tilde{\psi}| (X^m Z^n)^\dagger \\
&= \frac{\text{Tr}(|\tilde{\psi}\rangle\langle\tilde{\psi}|)}{k} \id_k \\
&= \frac{\id_k}{k}
\ee

Hence, 

\bb
\begin{aligned}
\cleancal{E}(|\psi\rangle\langle\psi|) &= \frac{1}{2} \begin{bmatrix} \frac{\id_k}{k} & \frac{\id_k}{k} F_k \\ F_k^\dagger \frac{\id_k}{k} & F_k^\dagger \frac{\id_k}{k} F_k \end{bmatrix} \\
&= \frac{1}{2k} \begin{bmatrix} \id_k & F_k \\ F_k^\dagger & \id_k \end{bmatrix}\\& \equiv \rho_{F_k}
\end{aligned}
\ee

Because our twirling channel equation $\cleancal{E}(|\psi\rangle\langle\psi|) = \rho_{F_k}$ is itself a valid convex decomposition of $\rho_{F_k}$ with uniform probabilities $p_{m,n} = \frac{1}{k^2}$, the coherence number must be less than or equal to the maximum rank found within this specific ensemble:
\begin{equation}
    \mathrm{CN}(\rho_{F_k}) \le \max_{m,n} \text{CR}(V_m W_n |\psi\rangle)
\end{equation}
    The operator $V_m W_n$ applies local generalised Pauli shifts $X^m$ and adds phase  $Z^n$ to the target registers, hence merely permuting  their positions, and applying  phases, respectively. Neither operation alters the total number of non-zero coefficients (the $L_0$-norm) defining the state's coherence rank. Hence, 
\begin{equation}
    \text{CR}(V_m W_n |\psi\rangle) = \text{CR}(|\psi\rangle) \quad \forall m,n
\end{equation}
Substituting this invariance back into our inequality yields:
\begin{equation}
    \mathrm{CN}(\rho_{F_k}) \le \text{CR}(|\psi\rangle)
\end{equation}
Since this statement holds true for \textit{any} state $|\psi\rangle$ chosen inside the support of $\rho_{F_k}$, the inequality must also hold true for the absolute minimum rank state within that subspace:
\begin{equation}
    \mathrm{CN}(\rho_{F_k}) \le \min_{|\psi\rangle \in \text{supp}(\rho_{F_k})} \text{CR}(|\psi\rangle)
\end{equation}
Pairing this upper bound with the well-established lower bound property $\mathrm{CN}(\rho) \ge \min_{|\psi\rangle \in \text{supp}(\rho)} \text{CR}(|\psi\rangle)$, the two definitions squeeze together perfectly, proving strict equality:
\bb
\mathrm{CN}(\rho_{F_k}) = \min_{|\psi\rangle \in \text{supp}(\rho_{F_k})} \text{CR}(|\psi\rangle)
\ee

Having established that $\mathrm{CN}(\rho_{F_k})$ is equivalent to the minimum pure-state coherence rank within the support of $\rho_{F_k}$, we must now analyse the structure of the states residing in this subspace.

Any pure state $|\psi\rangle \in \text{supp}(\rho_{F_k})$ is of the form:
\begin{equation}
    |\psi\rangle = \frac{1}{\sqrt{2}} \left( |0\rangle \otimes |a\rangle + |1\rangle \otimes F_k^\dagger |a\rangle \right)
\end{equation}
where $|a\rangle \in \mathbb{C}^k$ is an arbitrary normalised state vector in the target qudit space. 

The coherence rank (or $L_0$-norm) of a bipartite state in the computational basis is simply the total count of its non-zero entries. Because the control states $|0\rangle$ and $|1\rangle$ are perfectly orthogonal, the non-zero entries of the two branches do not overlap. The total coherence rank is strictly the sum of the non-zero entries of the two target-space blocks:
\begin{equation}
    \text{CR}(|\psi\rangle) = \||a\rangle\|_0 + \|F_k^\dagger |a\rangle\|_0
\end{equation}
Substituting this directly into our optimisation result, the global coherence number problem reduces to a purely mathematical search for the state $|a\rangle$ that minimises the combined support of itself and its inverse-Fourier transform:
\begin{equation}\label{min} 
\mathrm{CN}(\rho_{F_k}) = \min_{|a\rangle \in \mathbb{C}^k} \left( \||a\rangle\|_0 + \|F_k^\dagger |a\rangle\|_0 \right)
\end{equation}

To find the absolute minimum of $\||a\rangle\|_0 + \|F_k^\dagger |a\rangle\|_0$, we invoke the discrete uncertainty principle. The standard Donoho-Stark uncertainty relation~\cite{donoho-stark} dictates that $\||a\rangle\|_0 \times \|F_k^\dagger |a\rangle\|_0 \ge k$.  Tao showed that this can be improved further for cyclic groups of prime order $k$~\cite{Tao2005}, and established that $\||a\rangle\|_0 + \|F_k^\dagger |a\rangle\|_0\ge k+1$, thus obtaining a tighter lower bound. Later, Meshulam considers the generic case of composite $k$~\cite{meshulam}, refining this bound for functions over finite abelian groups. We explicitly state it for the sake of completeness.

\begin{Lemma}[\cite{meshulam}]
   Let $G$ be a finite abelian group of order $k$. For a complex valued function $f$ on $G$, let $\widehat{f}$ denote the Fourier transform of $f$.
Let $d_1 < d_2$ be two consecutive divisors of $k$. If $d_1 \leq x = |\text{supp}(f)| \leq d_2$ then
\bb
|\text{supp}(\widehat{f})| \geq \frac{k}{d_1 d_2} (d_1 + d_2 - x)
\ee 
\end{Lemma}

Let the support size of the state be $x = \||a\rangle\|_0$. Let $d_1$ and $d_2$ be two consecutive exact integer divisors of the dimension $k$ such that the support size falls in the interval $d_1 \le x \le d_2$. Thus according to Meshulam,  the support of the Fourier transform is bounded below by a strict linear interpolation between these divisors:
\begin{equation}
    \|F_k^\dagger |a\rangle\|_0 \ge \frac{k}{d_1 d_2}(d_1 + d_2 - x)
\end{equation}
We apply this inequality to our coherence rank objective function:
\begin{equation}
    \text{CR}(|\psi\rangle) \ge x + \frac{k}{d_1 d_2}(d_1 + d_2 - x)
\end{equation}
Grouping the terms by the variable $x$ yields a linear equation of the form $f(x) = mx + b$:
\begin{equation}
    \text{CR}(|\psi\rangle) \ge \left( 1 - \frac{k}{d_1 d_2} \right) x + \frac{k(d_1 + d_2)}{d_1 d_2}
\end{equation}
A fundamental property of any linear function defined on a closed interval $[d_1, d_2]$ is that it achieves its absolute minimum and maximum values exclusively at the domain boundaries. It is mathematically impossible for the minimum to occur at an intermediate integer strictly between the divisors.

Consequently, to find the global minimum of the coherence rank, we are completely justified in restricting our search to the domain boundaries where the support size $x$ is an exact divisor $r$ of $k$ (\textit{i.e.}\,$x = r$). If we evaluate the uncertainty sum precisely at an exact divisor $r$, the Fourier support bound simplifies directly to $k/r$. Thus, Meshulam's theorem  provides a rigorous, closed-form lower bound for the coherence number of the flower state:
\begin{equation}
    \mathrm{CN}(\rho_{F_k}) \ge \min_{r | k} \left( r + \frac{k}{r} \right)
\end{equation}
 which for prime $k$ boils down to $(k+1)$ as shown by Tao~\cite{Tao2005}.

To upgrade this lower bound to a strict equality, we must explicitly construct a valid quantum state $|a\rangle$ that perfectly achieves it. 

Let $r$ be the divisor of $k$ that minimises the expression $r + k/r$, and let $s$ be its complement such that $k = r \cdot s$. We define our test state $|a\rangle$ as:
\begin{equation}
    |a\rangle = \frac{1}{\sqrt{r}} \sum_{j=0}^{r-1} |j \cdot s\rangle
\end{equation}
By direct observation, this state possesses exactly $r$ non-zero amplitudes. Therefore, its standard support is exactly $\||a\rangle\|_0 = r$. 

Now, from the definition of the Fourier matrix,
\begin{equation}
    F_k^\dagger |y\rangle = \frac{1}{\sqrt{k}} \sum_{x=0}^{k-1} e^{-i \frac{2\pi}{k} x \cdot y} |x\rangle
\end{equation}
By linearity, applying $F_k^\dagger$ to $|a\rangle$ yields:
\begin{equation}
    F_k^\dagger |a\rangle = \frac{1}{\sqrt{k}} \sum_{x=0}^{k-1} \left( \frac{1}{\sqrt{r}} \sum_{j=0}^{r-1} e^{-i \frac{2\pi}{k} x\cdot (j \cdot s)} \right) |x\rangle
\end{equation}
We now substitute $k = r \cdot s$.  Thus,
\begin{equation}
    -i \frac{2\pi}{r \cdot s} x \cdot(j \cdot s) = -i \frac{2\pi}{r} x\cdot j
\end{equation}
The expression for the transformed state simplifies to:
\begin{equation}
    F_k^\dagger |a\rangle = \frac{1}{\sqrt{k \cdot r}} \sum_{x=0}^{k-1} \left( \sum_{j=0}^{r-1} \left( e^{-i \frac{2\pi}{r} x } \right)^j \right) |x\rangle
\end{equation}
Since the inner summation over $j$ is a geometric series of the $r$-th roots of unity, this sum evaluates to $r$ if the external index $x$ is an exact multiple of $r$. If $x$ is not a multiple of $r$, the phases destructively interfere and the sum evaluates to exactly $0$. Thus, $F_k^\dagger |a\rangle$ is non-zero exclusively at indices $x \in \{0, r, 2r, \dots, (s-1)r\}$. Because the dimension is $k$, there are exactly $k/r = s$ such non-zero entries. We conclude:
\begin{equation}
    \|F_k^\dagger |a\rangle\|_0 = s = \frac{k}{r}
\end{equation}
The total coherence rank of this explicitly constructed bipartite state is exactly $r + k/r$. Because this matches the absolute theoretical lower bound derived from Meshulam's theorem, we have proven that the bound is physically saturable.

Hence, we establish the exact  closed-form analytic solution for the coherence number of the flower state:
\begin{equation}
    \mathrm{CN}(\rho_{F_k}) = \min_{r | k} \left( r + \frac{k}{r} \right)
\end{equation}
To evaluate the asymptotic resource cost of the state, we must calculate the regularised coherence number $CN^\infty(\rho_{F_k}) = \lim_{n \to \infty} [CN(\rho_{F_k}^{\otimes n})]^{1/n}$. We shall evaluate it  via mathematical induction.

Any pure state $|\psi^{(n)}\rangle \in \text{supp}(\rho_{F_k}^{\otimes n})$ is a linear combination of the $n$-body eigenstates:$$|\psi^{(n)}\rangle = \sum_{m_1=0}^{k-1} \dots \sum_{m_n=0}^{k-1} A_{m_1, \dots, m_n} |\phi_{m_1}\rangle \otimes \dots \otimes |\phi_{m_n}\rangle$$

where \begin{equation}
    |\phi_m\rangle = \frac{1}{\sqrt{2}} |0, m\rangle + \frac{1}{\sqrt{2}} (\id \otimes F_k^\dagger)\ket{1}|m\rangle
\end{equation}
and $m \in \{0,1, \cdots, k-1\}$

Expanding the eigenstates separates the system into an $n$-qubit control register and an $n$-qudit target register. We group the terms by the $n$-bit control strings $\mathbf{y} \in \{0,1\}^n$:$$|\psi^{(n)}\rangle = \frac{1}{2^{n/2}} \sum_{\mathbf{y} \in \{0,1\}^n} |\mathbf{y}\rangle \otimes |A_{\mathbf{y}}\rangle$$, where each target branch $|A_{\mathbf{y}}\rangle$ is: $|A_{\mathbf{y}}\rangle = \left( \bigotimes_{i=1}^n (F_k^\dagger)^{y_i} \right) |A_{\mathbf{0}}\rangle$.

As the control states $|\mathbf{y}\rangle$ are strictly orthogonal, the total coherence rank of the $n$-copy state is the sum of the $L_0$-norms of the individual target branches:$$\text{rc}(|\psi^{(n)}\rangle) = \sum_{\mathbf{y} \in \{0,1\}^n} \||A_{\mathbf{y}}\rangle\|_0$$Let $U_k$ represent the absolute minimum bound of the discrete 1-qudit uncertainty principle, parametrised by exact divisors $r$: $U_k = \min_{r|k} (r + k/r)$. We proceed by induction to prove that $\text{rc}(|\psi^{(n)}\rangle) \ge U_k^n$. The base case $n=1$ is rigorously established by Meshulam's theorem. We form our inductive hypothesis: assume the bound holds for any $(n-1)$-copy state. Namely, for any non-zero $(n-1)$-qudit tensor $T$, the sum of its branch ranks satisfies:$$\sum_{\mathbf{y}' \in \{0,1\}^{n-1}} \left\| \left( \bigotimes_{i=1}^{n-1} (F_k^\dagger)^{y_i'} \right) |T\rangle \right\|_0 \ge U_k^{n-1} \cdot \id(|T\rangle \neq 0)$$

where $$\id(|T\rangle \neq 0) = 
\begin{cases} 
1 & \text{if } |T\rangle \text{ is a non-zero vector} \\
0 & \text{if } |T\rangle \text{ is the zero vector } (\text{all entries are } 0)
\end{cases}$$

Consider the $n$-copy target branches. We partition the control string into the first qubit decision $y_1 \in \{0,1\}$ and the remaining sequence $\mathbf{y}_{\text{rest}} \in \{0,1\}^{n-1}$. The total rank is:$$\text{rc}(|\psi^{(n)}\rangle) = \sum_{\mathbf{y}_{\text{rest}}} \Big[ \||A_{0, \mathbf{y}_{\text{rest}}}\rangle\|_0 + \||A_{1, \mathbf{y}_{\text{rest}}}\rangle\|_0 \Big]$$We slice the global $n$-qudit state along the computational basis of the first target register, indexed by $m_1$. Let $|T^{(m_1)}\rangle$ be the $(n-1)$-qudit state defined by fixing the first register of $|A_{\mathbf{0}}\rangle$ to $m_1$.The $y_1=0$ branch leaves the first register as it is. Thus this branch yields:$$|A_{0, \mathbf{y}_{\text{rest}}}\rangle = \sum_{m_1=0}^{k-1} |m_1\rangle \otimes \left[ \left( \bigotimes_{i=2}^n (F_k^\dagger)^{y_i} \right) |T^{(m_1)}\rangle \right] \equiv \sum_{m_1=0}^{k-1} |m_1\rangle \otimes |v_{0, m_1, \mathbf{y}_{\text{rest}}}\rangle$$The $y_1=1$ branch applies $F_k^\dagger$ to the first register before mapping the remaining qudits. Expanding the transform and grouping by the output basis $|m_1\rangle$ yields:$$|A_{1, \mathbf{y}_{\text{rest}}}\rangle = \sum_{m_1=0}^{k-1} |m_1\rangle \otimes \left[ \left( \bigotimes_{i=2}^n (F_k^\dagger)^{y_i} \right) \left( \sum_{j_1=0}^{k-1} \frac{1}{\sqrt{k}} e^{2\pi i m_1 j_1 / k} |T^{(j_1)}\rangle \right) \right]$$To formalise this, we define the state $|\tilde{T}^{(m_1)}\rangle = \sum_{j_1} \frac{1}{\sqrt{k}} e^{2\pi i m_1 j_1 / k} |T^{(j_1)}\rangle$. The second branch slices cleanly as:$$|A_{1, \mathbf{y}_{\text{rest}}}\rangle = \sum_{m_1=0}^{k-1} |m_1\rangle \otimes \left[ \left( \bigotimes_{i=2}^n (F_k^\dagger)^{y_i} \right) |\tilde{T}^{(m_1)}\rangle \right] \equiv \sum_{m_1=0}^{k-1} |m_1\rangle \otimes |v_{1, m_1, \mathbf{y}_{\text{rest}}}\rangle$$Because $\{|m_1\rangle\}$ are strictly orthogonal, we sum the $L_0$-norms across the slices and swap the order of summation between $m_1$ and $\mathbf{y}_{\text{rest}}$:$$\text{rc}(|\psi^{(n)}\rangle) = \sum_{m_1=0}^{k-1} \left[ \sum_{\mathbf{y}_{\text{rest}}} \||v_{0, m_1, \mathbf{y}_{\text{rest}}}\rangle\|_0 + \sum_{\mathbf{y}_{\text{rest}}} \||v_{1, m_1, \mathbf{y}_{\text{rest}}}\rangle\|_0 \right]$$By definition, $|v_{0, m_1, \mathbf{y}_{\text{rest}}}\rangle$ and $|v_{1, m_1, \mathbf{y}_{\text{rest}}}\rangle$ are exactly the $\mathbf{y}_{\text{rest}}$ branches generated from the initial $(n-1)$-qudit states $|T^{(m_1)}\rangle$ and $|\tilde{T}^{(m_1)}\rangle$. Applying our inductive hypothesis directly to these inner sums factors out $U_k^{n-1}$:
\bb
\text{rc}(|\psi^{(n)}\rangle) \ge U_k^{n-1} \sum_{m_1=0}^{k-1} \Big[ \id(|T^{(m_1)}\rangle \neq 0) + \id(|\tilde{T}^{(m_1)}\rangle \neq 0) \Big]\label{eq:71}
\ee

Since the global state $|\psi^{(n)}\rangle$ is non-zero, the initial tensor $|A_{\mathbf{0}}\rangle$ must contain at least one non-zero entry. Thus, there exists a specific, fixed coordinate string $\mathbf{c} \in \mathbb{Z}_k^{n-1}$ for the remaining registers such that extracting this column across all slices yields a non-zero 1-qudit state. We construct a 1-qudit target vector $|a\rangle$ by evaluating $|T^{(m_1)}\rangle$ at column $\mathbf{c}$ for all $m_1$:$$|a\rangle = \sum_{m_1=0}^{k-1} a_{m_1} |m_1\rangle \quad \text{where} \quad a_{m_1} = \langle \mathbf{c} | T^{(m_1)} \rangle$$Similarly, we extract column $\mathbf{c}$ from the transformed state $|\tilde{T}^{(m_1)}\rangle$ to form $|b\rangle$:$$|b\rangle = \sum_{m_1=0}^{k-1} b_{m_1} |m_1\rangle \quad \text{where} \quad b_{m_1} = \langle \mathbf{c} | \tilde{T}^{(m_1)} \rangle$$

By substituting the explicit definition of the state $|\tilde{T}^{(m_1)}\rangle$ into the evaluation of $b_{m_1}$, we reveal the relationship between the extracted columns:$$b_{m_1} = \langle \mathbf{c} | \left( \sum_{j_1=0}^{k-1} \frac{1}{\sqrt{k}} e^{2\pi i m_1 j_1 / k} |T^{(j_1)}\rangle \right) = \sum_{j_1=0}^{k-1} \frac{1}{\sqrt{k}} e^{2\pi i m_1 j_1 / k} \langle \mathbf{c} | T^{(j_1)} \rangle$$Substituting $a_{j_1} = \langle \mathbf{c} | T^{(j_1)} \rangle$, we get:$$b_{m_1} = \sum_{j_1=0}^{k-1} \frac{1}{\sqrt{k}} e^{2\pi i m_1 j_1 / k} a_{j_1}$$This confirms that the 1-qudit vectors form an exact inverse-Fourier pair: $|b\rangle = F_k^\dagger |a\rangle$.

Hence, \begin{align*}\id(|T^{(m_1)}\rangle \neq 0) &\ge \id(a_{m_1} \neq 0) \end{align*}
\begin{align*}
\id(|\tilde{T}^{(m_1)}\rangle \neq 0) &\ge \id(b_{m_1} \neq 0)
\end{align*}

\color{black}
Substituting these lower bounds into the inequality derived in (\ref{eq:71}) yields:$$\text{rc}(|\psi^{(n)}\rangle) \ge U_k^{n-1} \left( \sum_{m_1=0}^{k-1} \id(a_{m_1} \neq 0) + \sum_{m_1=0}^{k-1} \id(b_{m_1} \neq 0) \right)$$$$\text{rc}(|\psi^{(n)}\rangle) \ge U_k^{n-1} \Big( \||a\rangle\|_0 + \||b\rangle\|_0 \Big)$$Because $|a\rangle$ was strictly chosen to be non-zero, and $|b\rangle = F_k^\dagger |a\rangle$, the pair is subject to the discrete 1-qudit uncertainty principle: $\||a\rangle\|_0 + \||b\rangle\|_0 \ge U_k$.$$\text{rc}(|\psi^{(n)}\rangle) \ge U_k^{n-1} \times U_k = U_k^n$$This completes the mathematical induction, proving that the coherence rank of any valid $n$-copy decomposition is absolutely lower-bounded by $U_k^n$.

Having bounded the state strictly from below, we establish the matching upper bound by explicitly constructing a valid convex decomposition of $\rho_{F_k}^{\otimes n}$. The $n$-copy mixed state can be generated by applying the $n$-fold tensor product of the generalised twirling channel $\cleancal{E}$  to any pure state within its global support. We define an optimal $n$-copy pure state as the $n$-fold tensor product of the optimal single-copy state: $|\Psi\rangle = |\psi_{opt}\rangle^{\otimes n}$. Because the coherence rank of a tensor product pure state is strictly multiplicative, we have:$$\text{rc}(|\Psi\rangle) = (\text{rc}(|\psi_{opt}\rangle))^n = U_k^n$$

Feeding $|\Psi\rangle$ into the twirling channel $\cleancal{E}^{\otimes n}$ generates a valid ensemble decomposition of $\rho_{F_k}^{\otimes n}$ consisting entirely of uniformly weighted pure states of the form $\left( \bigotimes_{i=1}^n V_{m_i} W_{n_i} \right) |\Psi\rangle$. As already mentioned, the twirling operators $V_m$ and $W_n$ consist exclusively of computational basis shifts and phases. Neither of these operations alters the total number of non-zero amplitudes. Therefore, every pure state in this generated ensemble possesses a coherence rank of exactly $U_k^n$:$$\text{rc}\left( \left( \bigotimes_{i=1}^n V_{m_i} W_{n_i} \right) |\Psi\rangle \right) = \text{rc}(|\Psi\rangle) = U_k^n \quad \forall \, \mathbf{m}, \mathbf{n}$$By the definition of the coherence number, the existence of a valid decomposition whose maximum pure-state rank is $U_k^n$ guarantees that the coherence number of the mixed state is upper-bounded by this value:$$CN(\rho_{F_k}^{\otimes n}) \le U_k^n$$Because the analytical lower bound structurally guarantees $CN(\rho_{F_k}^{\otimes n}) \ge U_k^n$, and our physically constructed twirled ensemble achieves $CN(\rho_{F_k}^{\otimes n}) \le U_k^n$, the two bounds meet in strict equality for any finite $n$:$$CN(\rho_{F_k}^{\otimes n}) = U_k^n$$Taking the regularised limit as $n \to \infty$ removes the exponential scaling identically:$$CN^\infty(\rho_{F_k}) = \lim_{n \to \infty} \left( U_k^n \right)^{1/n} = U_k$$Substituting the explicit Meshulam divisor minimisation for $U_k$, we yield the exact, closed-form asymptotic coherence number for any dimension:$$CN^\infty(\rho_{F_k}) = \min_{r|k} \left( r + \frac{k}{r} \right)$$

\section{Appendix E: The one-shot one-way LOCC protocol for distilling generalised flower states}\label{appendixF}

Let us recall the generalised flower states which are associated with the following states:

\bb
\rho_V \coloneqq \frac{1}{2k} \begin{pmatrix} \id & V \\ V^\dagger & \id \end{pmatrix} , 
\ee
where $V^{-1\vphantom{\dagger}} = V^\dagger$ is any $k \times k$ unitary matrix, and the total dimension is $d = 2k$.

We shall now give an explicit one-way LOCC protocol that is able to distil $1$ ebit from this generic class of states, deterministically, while using only a single copy,
characterised by an arbitrary unitary $V$.

To see this, first note that the states $\rho_{V}$ and $\Omega_{V}$ are isometrically related. Let $\cleancal{L}$ be the isometry whose action on the standard basis is defined as:
\bb
\begin{aligned}
\cleancal{L} |i\rangle = |ii\rangle_{AB}
\end{aligned}
\ee
Hence, \bb
\begin{aligned}
\Omega_{V} = \cleancal{L}\rho_{V} \cleancal{L}^\dagger
\end{aligned}
\ee

Let $|\phi_m\rangle$ be an eigenvector of the state $\rho_{V}$ corresponding to eigenvalue $\lambda_m$. We observe that: 
\bb
\begin{aligned}
\Omega_{V} (\cleancal{L}^{\vphantom{\dagger}}\ket{\phi_m} ) &= (\cleancal{L}^{\vphantom{\dagger}} \rho_{V} \cleancal{L}^\dagger) (\cleancal{L}^{\vphantom{\dagger}} \ket{\phi_m}) \\
&= \cleancal{L}^{\vphantom{\dagger}} \rho_{V} (\cleancal{L}^\dagger \cleancal{L}^{\vphantom{\dagger}}) \ket{\phi_m} \\
&= \cleancal{L}^{\vphantom{\dagger}} \rho_{V} \ket{\phi_m} \\
&= \cleancal{L}^{\vphantom{\dagger}} (\lambda_m \ket{\phi_m}) \\
&= \lambda_m (\cleancal{L}^{\vphantom{\dagger}} \ket{\phi_m})
\end{aligned}
\ee
It shows that if $|\phi_m\rangle$ is an eigenvector of $\rho_{V}$ with eigenvalue $\lambda_m$, then the mapped state $\cleancal{L}|\phi_m\rangle$ is an eigenvector of $\Omega_{V}$ with the exact same eigenvalue $\lambda_m$.

For an arbitrary unitary $V$, from Corollary~\ref{cor3}, the eigenvectors of the state $\rho_V$ corresponding to the non-zero eigenvalue $\lambda_m = 1/k$ are:
\bb
\ket{\phi_m} = \frac{1}{\sqrt{2}} \ket{0, m} + \frac{1}{\sqrt{2}} \sum_{p=0}^{k-1} V^\dagger_{p,m} \ket{1, p}
\ee

Applying the isometry $\cleancal{L}$, the corresponding eigenvectors of $\Omega_V$ are:
\begin{equation}
    \cleancal{L}|\phi_m\rangle = \frac{1}{\sqrt{2}} |0, m\rangle_A|0, m\rangle_B + \frac{1}{\sqrt{2}} \sum_{p=0}^{k-1} V^\dagger_{p,m} |1, p\rangle_A|1, p\rangle_B
\end{equation}

Alice and Bob use the following one-way LOCC Kraus operators $A_x \otimes B_{x,y}$:
\bb
\begin{aligned}
A_x &= \frac{1}{\sqrt{k}} \sum_{i \in \{0,1\}} \sum_{j=0}^{k-1} \omega_k^{j\cdot x} \ketbraa{i}{i, j} \\
B_{x,y} &= \frac{1}{\sqrt{k}} \sum_{j=0}^{k-1} \omega_k^{j\cdot(y-x)} \ketbraa{0}{0, j} + \frac{1}{\sqrt{k}} \sum_{j=0}^{k-1} \sum_{p=0}^{k-1} \omega_k^{y\cdot j-p\cdot x}V_{j,p} \ketbraa{1}{1, p}
\end{aligned}
\ee

 Let us apply $A_x \otimes B_{x,y}$ to the first term of $\cleancal{L}|\phi_m\rangle$. 

\bb
(A_x \otimes B_{x,y}) \left( \frac{1}{\sqrt{2}} |0, m\rangle_A|0, m\rangle_B \right) 
&= \frac{1}{\sqrt{2}} (A_x|0, m\rangle_A) \otimes (B_{x,y}|0, m\rangle_B) \\
&= \frac{1}{\sqrt{2}} \left( \frac{1}{\sqrt{k}} \omega_k^{m \cdot x}|0\rangle_A \right) \otimes \left( \frac{1}{\sqrt{k}} \omega_k^{m \cdot(y-x)}|0\rangle_B \right) \\
&= \frac{1}{k\sqrt{2}} \omega_k^{m \cdot x + m \cdot y - m \cdot x} |00\rangle_{AB} = \frac{1}{k\sqrt{2}} \omega_k^{m \cdot y} |00\rangle_{AB}
\ee
Next, we apply it to the $|1\rangle$-branch:

\bb
(A_x \otimes B_{x,y}) \left( \frac{1}{\sqrt{2}} \sum_{p=0}^{k-1} V^\dagger_{p,m} \ket{1, p}_A\ket{1, p}_B \right) 
&= \frac{1}{\sqrt{2}} \sum_{p=0}^{k-1} V^\dagger_{p,m} (A_x\ket{1, p}_A) \otimes (B_{x,y}\ket{1, p}_B) \\
&= \frac{1}{\sqrt{2}} \sum_{p=0}^{k-1} V^\dagger_{p,m} \left( \frac{1}{\sqrt{k}} \omega_k^{p\cdot x}\ket{1}_A \right) \otimes \left( \frac{1}{\sqrt{k}} \sum_{j=0}^{k-1} \omega_k^{y\cdot j} V_{j,p} \omega_k^{-p\cdot x}\ket{1}_B \right) \\
&= \frac{1}{k\sqrt{2}} \sum_{j=0}^{k-1} \omega_k^{y\cdot j} \left( \sum_{p=0}^{k-1}  V_{j,p} V^\dagger_{p,m}  \right) \ket{11}_{AB}
\ee

Since
$\sum_{p=0}^{k-1} V^{\vphantom{\dagger}}_{j,p} V^\dagger_{p,m} = (V^{\vphantom{\dagger}} V^\dagger)_{j,m} = \delta_{j,m}$
we have
\bb
\begin{aligned}
(A_x \otimes B_{x,y}) \left( \frac{1}{\sqrt{2}} \sum_{p=0}^{k-1} V^\dagger_{p,m} |1, p\rangle_A|1, p\rangle_B \right) 
&= \frac{1}{k\sqrt{2}} \sum_{j=0}^{k-1} \omega_k^{y\cdot j} \delta_{j,m} |11\rangle_{AB} \\
&= \frac{1}{k\sqrt{2}} \omega_k^{m\cdot y} |11\rangle_{AB}
\end{aligned}
\ee
Hence, 
\bb
\begin{aligned}
(A_x \otimes B_{x,y})(\cleancal{L}|\psi_m\rangle) &= \frac{1}{k\sqrt{2}} \omega_k^{m\cdot y} \left( |00\rangle_{AB} + |11\rangle_{AB} \right) \\
&= \frac{1}{k} \omega_k^{m\cdot y} |\Phi^+\rangle
\end{aligned}
\ee
Therefore, the outer product yields:
\begin{equation}
    (A_x \otimes B_{x,y})(\cleancal{L}|\psi_m\rangle\langle\psi_m|\cleancal{L}^\dagger)(A_x \otimes B_{x,y})^\dagger = \frac{1}{k^2} |\Phi^+\rangle\langle\Phi^+|
\end{equation}

Summing over all outcomes $x, y,$ and eigenvectors $m$ yields exactly $1$ ebit of entanglement:
\bb
\begin{aligned}
\Lambda_{\text{1-LOCC}}(\Omega_V) &= \sum_{x,y,m} \frac{1}{k} \left( \frac{1}{k^2} |\Phi^+\rangle\langle\Phi^+| \right) \\
&= |\Phi^+\rangle\langle\Phi^+|
\end{aligned}
\ee
\end{document}